\documentclass[
final
 , nomarks
]{dmtcs-episciences}
\usepackage{silence}
\usepackage[utf8]{inputenc}

\usepackage{mathtools}

\usepackage[backend=biber, bibencoding=utf8, style=trad-abbrv, sorting=nty, sortcites, url=false, eprint=false]{biblatex}
\DeclareSourcemap{
  \maps[datatype=bibtex]{
    \map[overwrite]{ %
      \step[fieldsource=shortjournal,fieldtarget=journaltitle]
    }
  }  
}
\usepackage{csquotes}
\usepackage{xpatch}

\usepackage{amsthm}

\usepackage[capitalise, noabbrev]{cleveref}

\usepackage{ifthenx}

\usepackage{float}
\usepackage{subcaption}

\theoremstyle{plain}%
\newtheorem{theorem}{Theorem}[section]
\newtheorem{lemma}[theorem]{Lemma}

\newtheorem{conj}{Conjecture}
\newtheorem*{corollary}{Corollary}
\theoremstyle{definition}
\newtheorem{definition}[theorem]{Definition}
\newtheorem{example}[theorem]{Example}

\theoremstyle{remark}
\newtheorem*{rem}{Remark}

\newtheorem{case}{Case}
\newtheorem*{obs}{Observation}

\newcommand{\et}[0]{\text{ and }}
\newcommand{\sC}[0]{\mathfrak{C}}
\newcommand{\RL}[0]{\preceq}
\newcommand{\NN}[0]{\mathbb{N^+}}
\newcommand{\comp}[1]{composition of \ensuremath{#1}}
\newcommand{\nst}[0]{\textsuperscript{st}}
\newcommand{\nth}[0]{\textsuperscript{th}}

\newcommand{\range}[2]{%
#1,\dotsc,#2}

\newcommand{\setrange}[2]{%
\left\{\range{#1}{#2}\right\}}

\newcommand{\nrange}[1]{%
\left[#1\right]}

\newcommand{\setsize}[1]{%
\left|#1\right|}

\newcommand{\seq}[2][k]{
#2_1#2_2\dots#2_#1
}

\DeclareMathOperator{\supp}{supp}
\DeclareMathOperator{\Occ}{Occ}
\DeclareMathOperator{\std}{st}
\DeclareMathOperator{\Pro}{Pro}
\DeclareMathOperator{\rev}{rev}

\usepackage{tikz}
\usepackage{xinttools}
\usetikzlibrary{automata, positioning, arrows, patterns, decorations.pathreplacing}

\tikzset{
        >=stealth',
        node distance = 2.5cm,
        every state/.style={thick, fill=gray!10},
        initial text=$ $
        }

\newcommand{\plotskylineoutline}[1]{%
  \foreach \j [count=\i] in {#1} {
    \ifthenelse{\j > 0}{
    \draw [thick] (\i-1,0) rectangle ++(1,\j);}
    {}
  };
  }

\newcommand{\plotskyline}[1]{%
  \foreach \j [count=\i] in {#1} {
  \ifthenelse{\j > 0}{
    \foreach \k in {1,...,\j} {
    \draw [thick] (\i-1,\k-1) rectangle ++(1,1);
    };
    }{}
  };
  }
  
\newcommand{\plotskylineshaded}[1]{%
  \foreach \j [count=\i] in {#1} {
  \ifthenelse{\j > 0}{
    \foreach \k in {1,...,\j} {
    \draw [lightgray, fill=lightgray, thick] (\i-1,\k-1) rectangle ++(1,1);
    };}{}
  };
  }

\newcommand{\plotskylineshadedboxes}[1]{
  \foreach \x/\y in {#1}{
  \draw [lightgray, fill=lightgray, thick] (\x-1,\y-1) rectangle ++(1,1);
  };
  }

\newcommand{\drawthegrid}[1]{%
\draw (0.01,0.01) grid (#1+0.99,#1+0.99);
}

\newcommand{\drawthepoints}[2]{%
\foreach[count=\x] \y in {#1}
\filldraw (\x,\y) circle (#2 pt);
}

\newcommand{\clpattern}[3][5]{%
    \pgfmathsetmacro\circlesize{#1+4}
    \drawthegrid{\xintNthElt{0}{\xintCSVtoList {#3}}}
    \drawthepoints{#3}{#1}
    \foreach \x in {#2}
    \draw (\x,\xintNthElt{\x}{\xintCSVtoList {#3}}) circle (\circlesize pt);
}

\title{Prolific Compositions}
\author{Michael Albert \and Murray Tannock}
\affiliation{Department of Computer Science, University of Otago, Dunedin, New Zealand}

\received{2019-4-12}

\revised{2019-10-28}

\accepted{2019-11-13}

\publicationdetails{21}{2019}{2}{10}{5373}
\begin{document}
\maketitle

\begin{abstract}
  Under what circumstances might every extension of a combinatorial structure contain
  more copies of another one than the original did? This property, which we call \emph{prolificity}, holds universally
  in some cases (e.g., finite linear orders) and only trivially in others (e.g., permutations).
  Integer compositions, or equivalently layered permutations, provide a middle ground.
  In that setting, there are prolific compositions for a given pattern if and only if that pattern begins and ends
  with 1. For each pattern, there is an easily constructed automaton that recognises prolific compositions for
  that pattern. Some instances where there is a unique minimal prolific composition for a pattern
  are classified.
\end{abstract}

\section{Introduction}
In combinatorics we are often interested in the ways that one structure, a pattern, can occur inside
another, the text. There are many different ways to study the nature of occurrences. At the heart of
the study of permutation patterns is the notion of pattern avoidance, studying permutations which
contain no occurrences of given pattern. Early papers in this area include
\textcites{simionRestrictedPermutations1985,stankovaForbiddenSubsequences1994}. A survey of the
current state of research on classical permutation patterns can be found in \textcite{vatterPermutationClasses2015}.
Whilst considering permutation patterns
\textcites{fulmekEnumerationPermutationsContaining2003,bonaPermutationsOneTwo1998} are amongst those
who have considered texts containing a prescribed number of occurrences.
\textcite{bonaAbsencePatternOccurrences2010} also examined the case  where some patterns definitely
occur, but other patterns are absent. Asymptotic statistics on occurrences of certain patterns as
well as the distribution of occurrences of patterns amongst texts of different lengths have been
studied by \textcites{jansonPatternsRandomPermutations2017,jansonAsymptoticStatisticsNumber2015}. It
is also possible to study the texts that permit the highest number of occurrences, or the packing
density as presented by \textcites{albertPackingDensitiesPermutations2002} and others.
\textcite{kuszmaulFastAlgorithmsFinding2018} has attempted to find efficient algorithmic methods to
count the number of occurrences of a pattern in each member of a set of texts. Some of these
problems have also been studied in the context of words, with
\textcite{bursteinCountingOccurrencesSubword2003} examining prescribed counts of patterns and
\textcite{bursteinPackingPatternsWords2002} studying packing density of words under subword order.

Another approach we can take, and the one considered in this paper is to examine how the number of
occurrences of the pattern can change as we change the  containing structure. If we add a new
element to the text then the number of occurrences of the pattern must either stay the same, or
increase. We will investigate the combinations of patterns and texts having the property that
regardless of how a new element is added to the text the number of occurrences of the pattern
increases.

For example, consider the case of finite linear orders. Since, up to isomorphism, there is only one
such structure of any given size, any \(k\) elements of a linear order on \(n\) elements, represents
an occurrence of the \(k\) element order so the number of occurrences is \(\binom{n}{k}\) and no
matter how we add a new element the number of \(k\) element suborders increases, provided that \(n
\ge k > 0\), since \begin{equation*} \binom{n+1}{k} > \binom{n}{k} \end{equation*}

This gives rise to the following definition

\begin{definition}
    A text \(\tau\), is \emph{prolific} for a pattern \(\pi\), or \emph{\(\pi\)-prolific}, if every
    proper extension of \(\tau\) contains more occurrences of \(\pi\) than \(\tau\).
\end{definition}

\begin{obs}
    If \(\tau\) is prolific for \(\pi\) and \(\tau\) is contained in \(\nu\) then \(\nu\) is
    prolific for \(\pi\) since any extension of \(\nu\) contains extensions of \(\tau\)
    and therefore also new occurrences of \(\pi\). 
    Thus, the set of texts defined by the property of being prolific for \(\pi\) is
    upwards closed set with respect to containment order.
 \end{obs}

In contrast to finite linear orders, among the set of all permutations only the
singleton permutation has any prolific texts. This is because any non-singleton
permutation, \(\pi\),  cannot end with both its maximum and its minimum element. If \(\pi\)
does not end with its maximum then extending any other permutation by a new
maximum element at its end creates no new copies of \(\pi\). In the other case, extension
by a new minimum element has the same effect.

Linear orders (which could be thought of simply as monotone permutations) and the set of all
permutations lie at opposite ends of the scale in terms of what patterns have prolific texts and
what those texts are. Neither one allows for an interesting study of prolificity in general. To see 
whether this concept is of interest at all we need to demonstrate the existence of a middle ground
between these two extremes. It turns out that integer compositions (or in terms of permutations,
layered permutations) occupy part of that middle ground.

\section{Basic Definitions}

Let \(n\) be a positive integer. A sequence of positive integers whose sum is \(n\) is called a
\emph{\comp{n}}. We can display any composition graphically as a \emph{skyline diagram} as shown in
\cref{fig:skyline}. Denote by \(\sC_n\) all the compositions of \(n\) and let \(\sC\) denote the set
of all compositions.

\begin{figure}[htb]
\centering
  \begin{tikzpicture}[scale=0.5]
	\plotskyline{1,3,2,4,2,1,1};
	\end{tikzpicture}
	\caption{The skyline diagram of a composition}
	\label{fig:skyline}
\end{figure}
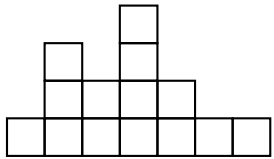

A \comp{n} can also be thought of as a partition of the set \(\nrange{n} = \setrange{1}{n}\) whose
parts form intervals. For example, the composition with part sizes \(2,1 \et 2\) corresponds to
partition of \(\nrange{5}\) into the sets \(\{1,2\}, \{3\}, \et \{4,5\}\). The correspondence
between the partition view and the skyline diagrams can be seen in \cref{fig:comprepr}. As the skyline diagram is more
compact we will use this view throughout.

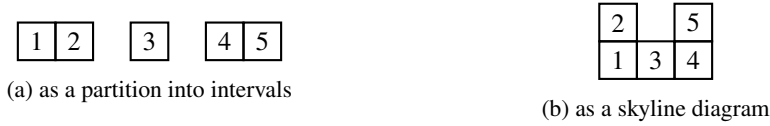
\begin{figure}[htb]
\centering
\begin{subfigure}{.45\textwidth}
\centering
\begin{tikzpicture}[scale=0.5]
  \plotskyline{1,1,0,1,0,1,1};
  \draw (0.5,0.5) node {1}
        (1.5,0.5) node {2}
        (3.5,0.5) node {3}
        (5.5,0.5) node {4}
        (6.5,0.5) node {5};
\end{tikzpicture}
\caption{as a partition into intervals}
\end{subfigure}
\begin{subfigure}{.45\textwidth}
\centering
\begin{tikzpicture}[scale=0.5]
  \plotskyline{2,1,2};
  \draw (0.5,0.5) node {1}
        (0.5,1.5) node {2}
        (1.5,0.5) node {3}
        (2.5,0.5) node {4}
        (2.5,1.5) node {5};
\end{tikzpicture}
\caption{as a skyline diagram}
\end{subfigure}
\caption{Different representations of a composition}\label{fig:comprepr}
\end{figure}

Of course there is also a correspondence between compositions of \(n\) and layered permutations
on \(n\) elements. When a composition is written as a partition into intervals, by writing the elements of each interval
from largest to smallest, and writing the intervals in order we immediately obtain the corresponding layered
permutation, as shown in \cref{ex:com2layered}.

\begin{example}
    \label{ex:com2layered}
    The composition \((\{1,2\}, \{3\},\{4,5\},\{6,7,8,9\})\) maps to the layered permutation \(213549876\).
    
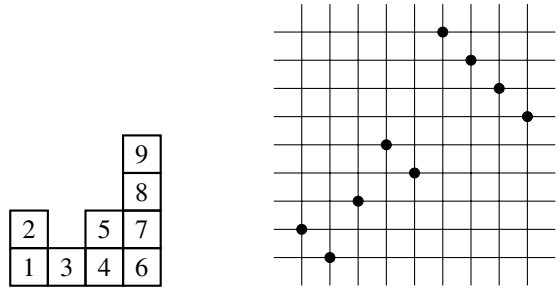
\begin{center}
    \begin{tikzpicture}[scale=0.5]
    \plotskyline{2,1,2,4}
    \draw (0.5,0.5) node {1}
        (0.5,1.5) node {2}
        (1.5,0.5) node {3}
        (2.5,0.5) node {4}
        (2.5,1.5) node {5}
        (3.5,0.5) node {6}
        (3.5,1.5) node {7}
        (3.5,2.5) node {8}
        (3.5,3.5) node {9}
        ;
        \begin{scope}[shift={(7,0)}, scale=0.75]
            \clpattern{}{2,1,3,5,4,9,8,7,6}
        \end{scope}
    \end{tikzpicture}
\end{center}
\end{example} 

Although we do not make use of this correspondence in the following work it is
part of the underlying motivation: to determine what patterns have prolific texts, and what those texts
are in various permutation classes.

An element of a composition is a member of the underlying set and in the skyline diagram, each
individual square also represents an element of the underlying set. The size of a composition,
\(v\), is the size of its underlying set and is denoted \(|v|\); furthermore, a \emph{part of size
\(x\)} is a part in the composition where the corresponding interval contains \(x\) integers.

Compositions can also be represented as words over \(\NN\) where each letter is the size
of the corresponding part in the composition. This is the most convenient representation in text
so we will write compositions in concatenative
notation by just listing the sizes of their parts, the composition in \cref{fig:skyline} is
therefore written as \(1324211\). When we consider compositions as words denote partial words by
greek letters and individual parts by numerals or roman letters.

A new elements can be \emph{inserted} into a composition by
\begin{enumerate}
	\item increasing the size of some part by one, or
	\item creating a new part of size one adjacent to some existing part
\end{enumerate}
There are therefore, up to isomorphism, \(k+(k+1)\) ways of inserting a new element into a
composition with \(k\) parts. Any composition obtained by inserting one or more elements into the composition \(p\) are
called the \emph{extensions} of \(p\).

Consider how to find an \emph{embedding}, of one \comp{n} inside another \comp{m}. In the context of
partitions, we can say that an embedding is an strictly order-preserving injection from
\(\nrange{n}\) to \(\nrange{m}\) such that two elements belong to the same part of the \comp{n} if
and only if their images belong to the same part of the \comp{n}. \cref{fig:compocc} shows an
embedding of the composition  \(122321\) in the composition \(13224211\).

\begin{figure}[htb]
\centering
    \begin{tikzpicture}[scale=0.5]
  \plotskylineshadedboxes{1/1,2/1,2/3,4/1,4/2,5/1,5/2,5/4,6/1,6/2,8/1}
	\plotskyline{1,3,2,2,4,2,1,1};
	\end{tikzpicture}
	\caption{An occurrence of \(122321\) in \(13224211\).}
	\label{fig:compocc}
\end{figure}
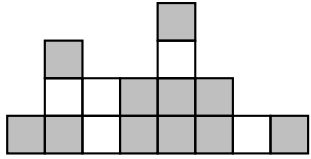

When we consider compositions as a sequence whose sum is \(n\) then a composition, \(u\), with
\(\ell\) parts has an embedding into a composition, \(v\), with \(k\) parts if we can find a subset
of \(\ell\) parts of the composition \(v\) such that each part of this subset is at least as large
as the corresponding part in \(u\). This is formalised in the following definition.

\begin{definition}\label{def:compcont}
  Given two compositions \(u=\seq{u} \et v=\seq[n]{v}\) (where \(k \le n\)) let \(\RL\) be the
  binary relation on the set of compositions such that \( u\RL v\) if there exists a set of \(k\)
  indices \(I=\setrange{i_1}{i_k}\) such that \(u_\ell \le v_{i_\ell}\) for all \(\ell \in
  \setrange{1}{k}\).
\end{definition}

An occurrence of a composition \(u\) with \(k\) parts in a composition \(v\) under the relation
\(\RL\) is a way of  selecting \(k\) parts of \(v\) and then from each of those parts \(i\) choosing
\(u_i\) elements from those parts.

When looking for an occurrence of a composition \(u =
\seq{u}\) inside a composition \(v\) we can match the first part of \(u_1\) into the leftmost
part of \(v_{i_1}\) of \(v\) with size greater than \(u_1\), we can then match \(u_2\) into the left most
sufficiently large part of \(v\) after \(v_{i_1}\) and so on until we have seen an occurrence of \(u\) (or run out of parts in \(v\)). This
\emph{greedy} approach allows us to check whether \(v\) contains \(u\) in time proportional to the
length of \(v\) in the worst case. The existence of a leftmost occurrence of \(u\) in \(v\), or of a rightmost
occurrence, or similar considerations based around the greedy algorithm is a critical part of many
of the proofs below.

The \emph{support} of an occurrence of \(u\) in \(v\) is the  set of indices that correspond to the
chosen parts, for example, the support of the occurrence indicated in  \cref{fig:compocc} is
\(\{1,2,4,5,6,8\}\).

\begin{definition}
  Given a composition \(u\) with \(k\) parts, and a set of indices  \(X \subseteq \setrange{1}{k}\),
  the \emph{subcomposition of  \(u\) at \(X\)}, denoted \(u_X\) is the composition formed by  taking
  the parts of \(u\) at the indices selected by \(X\).
\end{definition}

Given a composition \(u\) with \(k\) parts and a composition \(v\) with \(\ell\) parts  where \(k\le
\ell\), the set of \emph{supports of \(u\) in \(v\)}, denoted \(\supp(u,v)\), is the those
\(X\in\binom{\nrange{\ell}}{k}\) such that \(u \RL v_X\).

\begin{example}
  The supports of occurrences of \(u = 122321\) in \(v = 13224211\) can be found by noting that the
  part of size \(3\) in \(u\) must occur in the part of size \(4\) in \(v\). We can therefore take
  any \(3\) of the first \(4\) parts, the part of size \(4\) and the part of size \(2\) that follows
  it as well as either of the last two parts. The set of supports of \(u\) in \(v\) therefore has
  \(8\) elements.
\end{example}

Denote by \(\Occ(u,v)\) the number of occurrences of \(u\) inside
the compositions \(v\). For this relation, each part of size \(k\) in \(v\) that
matches a part of size \(\ell\) in \(u\) contributes a factor of \(\binom{\ell}{k}\)
to the number of occurrences. That is

\begin{equation}
  \Occ_\mathcal{L}(u,v) = \sum_{X\in \binom{\nrange{\ell}}{k}}
  \prod_{j\in \nrange{k}}\binom{v_{X_j}}{u_j} \label{eqn:occ}
\end{equation}

Note that if the set of indices chosen is not in the set of supports then the that particular term of the sum will be equal to \(0\).

If the union of the set of supports of occurrences of a composition \(u\) in a
composition \(v\) is equal to the set of all the indices of \(v\) then we say
that \(v\) is \emph{covered} by \(u\).

The reverse of a composition \(u\), \(\rev(u)\), is the \comp{|u|} that has the same parts as \(u\)
but the order of the parts are reversed.
\begin{rem}
  The number of occurrences of \(u\) in \(v\) is the same as the number of occurrences of
  \(\rev(u)\) in \(\rev(v)\). In fact \(\rev: \sC \to \sC\) is an automorphism of the collection of
  compositions as an ordered set. This symmetry will sometimes be implicitly used in our proofs.
\end{rem}

Given a composition \(u\) with \(k\) parts and an index \(i\) in \(\nrange{k}\) the \emph{prefix of
\(u\) up to \(i\)}  is the subcomposition of \(u\) at \(I=\setrange{1}{i-1}\). Similarly, the
\emph{suffix of \(u\) after \(i\)}  is the subcomposition of \(u\) at \(I=\setrange{i+1}{k}\). Note
that neither of the \emph{prefix of \(u\) up to \(i\)} or the \emph{suffix of \(u\) after \(i\)}
contain the \(i\)\nth{} part of the composition.

\section{Prolific Compositions}

Recalling the definition of being prolific in the context of compositions: a composition \(v\) of
size \(n\) is \emph{prolific} for a composition \(u\), or \emph{\(u\)-prolific}, if for all proper
extensions, \(v^\prime\) of \(v\), the number of occurrences of \(u\) in \(v^\prime\) is strictly
greater than the number of occurrences of \(u\) in \(v\). To determine whether \(v\) is
\(u\)-prolific it suffices to consider extensions, \(v^\prime\), of \(v\) with \(\setsize{v^\prime}
= \setsize{v} + 1\).  Denote the set of all \(u\)-prolific compositions as \(\Pro(u)\).

\begin{rem}
  A composition \(v\) is in \(\Pro(u)\) if and only if \(\rev(v)\) is in \(\Pro(\rev(u))\).
\end{rem}

\begin{theorem}\label{thm:startend}
  If \(u= \seq{u}\), then \(\Pro(u)\) is non-empty if and only if
  \(u_1 = u_k = 1\)
\end{theorem}

\begin{proof}
If \(u_1 \neq 1\) and \(v\) is any composition then \(1v\) contains no additional occurrences of
\(u\) compared to \(v\). Similarly if \(u_k \neq 1\) then \(v1\) contains no more occurrences of
\(u\) than \(v\) does. So, in both cases \(\Pro(u) = \emptyset\).

Conversely, if \(u=1\alpha1\) then \(v = 1\alpha\alpha{}1\) is prolific. If a new element is
inserted into \(v\) before the end of  the first \(\alpha\) then it creates a new occurrence of
\(u\) using the new element, the second \(\alpha\) and the final part. While if a new element is
inserted into \(v\) after the end of the first \(\alpha\) then it creates a new occurrence of \(u\)
using the first one, the first \(\alpha\), and the new element.
\end{proof}

\begin{lemma}\label{lem:oneocc}
  If a composition \(v\) is \(u\)-prolific, then it must contain at least one occurrence of \(u\).
\end{lemma}

\begin{proof}
  Suppose there were a composition \(v\) that was \(u\)-prolific but did not contain an occurrence
  of \(u\). Since the last part of \(u\) must be of size \(1\) then increasing the size of the last
  part of \(v\) would not create any occurrences of \(u\), so therefore \(v\) could not
  have been \(u\)-prolific. Any \(u\)-prolific composition must therefore contain at least one
  occurrence of \(u\).
\end{proof}

\begin{rem}
  Since a \(u\)-prolific composition must have at least one occurrence of \(u\) then it must have at
  least the same number of parts as \(u\).
\end{rem}

In fact the following lemma extends strengthens the statement made in \cref{lem:oneocc} considerably.
\begin{theorem}\label{thm:coveringprolific}
  If a composition \(v\) is \(u\)-prolific, then it must be covered by \(u\).
\end{theorem}

\begin{proof}
  Suppose there existed a composition \(v\) than belonged to \(\Pro(u)\)
  but was not covered by copies of \(u\). Choose a part \(X\) of \(v\), with size \(x\),
  that does not occur in the support of any occurrence of \(u\). Write \(v\) as
  \(\alpha x\beta\).
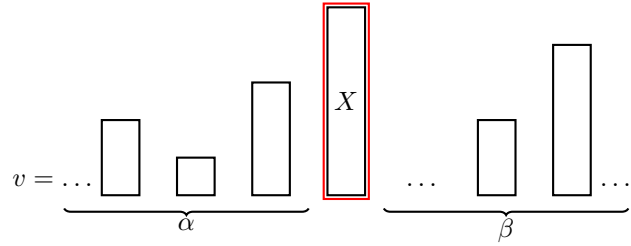
  \begin{center}
    \begin{tikzpicture}[scale=0.5]
      \node[above left] at (0,0) {\(v = \dots\)};
      \plotskylineoutline{2,0,1,0,3,0,5,0,0,0,2,0,4}
      \node[above right] at (13,0){\(\dots\)};
      \node at (6.5,2.5) {\(X\)};
      \node[above] at (8.5,0){\(\dots\)};
      \draw [red, thick] (5.9,-0.1) rectangle ++(1.2,5.2);
      \draw [
    thick,
    decoration={
        brace,
        mirror,
        raise=5
    },
    decorate
] (-1,0) -- (5.5,0)
node [pos=0.5,anchor=north,yshift=-5] {\(\alpha\)};

      \draw [
    thick,
    decoration={
        brace,
        mirror,
        raise=5
    },
    decorate
] (7.5,0) -- (14,0)
node [pos=0.5,anchor=north,yshift=-5] {\(\beta\)};
    \end{tikzpicture}
  \end{center}

  Consider the extension of \(v\) formed by inserting a new part of size one immediately after
  \(X\). Since \(v\) is \(u\)-prolific this must create a new occurrence of \(u\). If the new part
  were to play the role of the first part of \(u\) in the new occurrence then there would have
  existed a number of occurrences of \(u\) that used the same parts as the new occurrence, but used
  an element of the part \(X\) instead of the new part, and \(X\) would be covered. Therefore, the
  new part must play the role of the last part, or some internal part, of \(u\). Furthermore in the
  new occurrence of \(u\) must use the part \(X\) as otherwise any element of \(X\) could play the
  role that the inserted element plays and \(X\) would be covered.

  Now, suppose that the maximum prefix of \(u\) in \(\alpha X\) were the same as that in \(\alpha\).
  Then whatever parts of \(u\) that were used in the occurrence could have been found entirely in
  \(\alpha\), contradicting the observation that \(X\) must be used. So the maximum prefix of \(u\)
  in \(\alpha X\) is strictly longer than that in \(\alpha\). Hence when greedily finding an
  occurrence of \(u\) in \(v\) the part \(X\) will be used and this contradicts with the assumption
  that \(X\) was not covered.

\end{proof}

\begin{corollary}
  In order to discover whether a composition \(v\) is \(u\)-prolific we need
  only check that:
  \begin{itemize}
    \item \(v\) is  covered by \(u\), and
    \item for all factorisations of \(v = \alpha xy\beta\) with \(x,y > 1\) there is an occurrence of
  \(u\) in \(\alpha x1y\beta\) that uses the inserted \(1\).
  \end{itemize}
\end{corollary}
\begin{proof}
  If \(v\) is covered then increasing any part increases at least one term of the sum in
  \cref{eqn:occ} and therefore increases the number of occurrences. Inserting any new singleton
  parts adjacent to an existing part of size \(1\) in \(v\) need not be considered since the new part
  can be substituted for the existing part in any occurrence that used the existing part. Therefore
  it suffices to consider inserting new parts of size \(1\) between parts of size strictly greater
  than \(1\).
\end{proof}

Now consider the problem of how to efficiently determine that \( v \) is prolific for \( u \) (assuming
that \( u \) starts and ends with 1). The following description will provide the basis for an automaton that
recognises this property. That is, the automaton will accept exactly those words over \( \mathbb{N} \) that are \( u \)-prolific.
Each letter of such a word corresponds to processing a part in a composition so the non-accepting states
can be thought of as encoding the conditions that must be satisfied by the remaining suffix of the word being
processed in order that it should be \( u \)-prolific. What are these suffix conditions?

First consider the condition that inserting a 1 after any prefix of \( v \) should create a new occurrence
of \(u \). The prefix \(v_1 v_2 \dots v_i\) contains some maximal prefix \(\alpha\) of \(u\). Then there is a maximal \(\beta\)
such that \(\alpha 1\) contains the prefix \(\beta 1\) of \(u\). This might be \(\alpha 1\) itself, but if \(\alpha\) is not
followed immediately by a \(1\) in \(u\) then \(\beta\) will be the maximal prefix of \(\alpha\) that is followed by a \(1\) in \(u\). Such
a \(\beta\) will always exist (though it may be empty) since \(u\) begins with a \(1\). Now write \(u = \beta 1 \gamma\). The requirement
that the new 1 after \(v_i\) creates a new occurrence of \(u\) is precisely the requirement that the remainder of \(v\) should contain
an occurrence of \(\gamma\). This is the local suffix requirement at $v_i$.

At the point where \(v_i\) is being considered there was already some existing suffix requirement. If \(v_i\) is greater than
or equal to the first character of that suffix then this requirement is reduced in length by one, otherwise it stays the
same. The final suffix requirement after considering \( v_i \)  is then the longer of the previous (possibly modified) one,
and the local suffix requirement at \( v_i \).

The initial suffix requirement (before we process any characters at all) is all of \( u \). This ensures that,
as soon as the suffix requirement becomes empty we will have constructed a word containing an occurrence of \( u \). So, when
the suffix requirement is non-empty we have either not found an occurrence of \( u \) or there remain positions where inserting
a 1 would not create a copy of \( u \). To that point, the word we have processed is certainly not \( u \)-prolific. But, as soon as the suffix
requirement is empty these two conditions are both satisfied and we will show below that all such \( v \) are \( u \)-prolific.

\begin{example}
Given the composition \(u = 1213221\) and \(v = 15512443221\) we can associate the \(u\)-prefix, \(u\)-suffix requirement pairs show in \cref{t:prefixsuffix}.

\begin{table}[hbt]
\centering
\caption{Maximal prefixes and suffix requirement pairs for \ensuremath{u = 1213221} and \ensuremath{v = 15512443221}} \label{t:prefixsuffix}
\begin{tabular}{c|c|c|c}
\(i\) & \(v_i\) & \(u\)-prefix & \(u\)-suffix \\\hline
1 & 1 & 1 & 213221 \\
2 & 5 & 12 & 13221 \\
3 & 5 & 121 & 3221 \\
4 & 1 & 121 & 3221 \\
5 & 2 & 121 & 3221 \\
6 & 4 & 1213 & 3221 \\
7 & 4 & 12132 & 3221 \\
8 & 3 & 121322 & 221 \\
9 & 2 & 1213221 & 21 \\
10 & 2 & 1213221 & 1\\
11 & 1 & 1213221 & \(\varepsilon\)\\
\end{tabular}
\end{table}
\end{example}

The preceding discussion proves:

\begin{theorem}\label{thm:prefixsuffix}
  Given a composition \(v\) and a composition \(u\),  \(v\) is \(u\)-prolific if and only if
  the suffix requirements for the last part of \(v\) are empty.
\end{theorem}

\begin{proof}
We have already argued that if the suffix requirement is non-empty then \( v \) is not \( u \)-prolific.

Suppose that \( v \) finishes with an empty suffix requirement. Such a \( v \) contains occurrences of \(u \)
and the insertion of a 1 at any point creates new occurrences of \( u \). It remains only to show that \( v \)
is covered by \( u \).

Let \( X \) be any part of \( v \) and write \( v = \alpha X \beta \). We consider two cases: the maximal prefix
of \( u \) in \( \alpha X \) is the same as that in \( \alpha \), or it is longer. Because the suffix condition is eventually
empty we know that there is a new occurrence of \( u \) in \( \alpha X 1 \beta \) which necessarily uses the 1
and corresponds to a factorisation \( u = \gamma 1 \tau \) where \( \gamma \) occurs in \( \alpha X \).

In the first case, \( \gamma \) also occurs in \( \alpha \) and so we could use any element of \( X \) in place of the 1, i.e.,
that part is covered.

But in the second case we already know that \( X \) is covered. For the fact that the maximal prefix of
\( u \) in \( \alpha X \) is longer than that in \( \alpha \) implies that \( X \) will be part of the leftmost occurrence of
\( u \) in \( v \) (and we know such an occurrence exists).

Thus every part of \( v \) is covered and we can conclude that \( v \) is \( u \)-prolific.

\end{proof}

As already foreshadowed, \cref{thm:prefixsuffix} allows the construction an automaton that recognises \(u\)-prolific
compositions. In this automaton the states are given by pairs \((p,s)\), where \(p\) represents the
length of the prefix seen and \(s\) is the length of suffix requirement needed when inserting a new
element after the current partial composition. There exists a transition on an interval \(I=[a,b]\)
between states \((p,s)\) and \((p^\prime, s^\prime)\) if appending any value in the interval \(I\)
to a word seen at \((p,s)\) causes the prefix of \(u\) seen to become \(p^\prime\) symbols long and
causes the length of the suffix requirement to become \(s^\prime\). If a state is unreachable we
omit it from the final automaton. Note that \(p^\prime = p \) or \(p+1 \et s^\prime = s\) or
\(s-1\).

\begin{example}
  The automaton for recognising \(1441\)-prolific compositions.
  
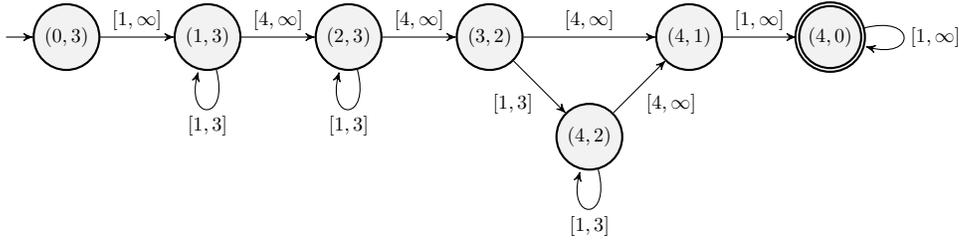
\begin{center}
\begin{tikzpicture}[scale=0.75, every node/.style={scale=0.75}]
    \node[state, initial] (q01) {\((0,3)\)};
    \node[state, right of=q01] (q11) {\((1,3)\)};
    \node[state, right of=q11] (q21) {\((2,3)\)};
    \node[state,  right of=q21] (q32) {\((3,2)\)};
    \node[state, below right of=q32] (q42) {\((4,2)\)};
    \node[state, above right of=q42] (q43) {\((4,1)\)};
    \node[state, accepting , right of=q43] (q44) {\((4,0)\)};

    \draw [->] (q01)  edge[above] node{\([1,\infty]\)} (q11)
          (q11) edge[loop below] node{\([1,3]\)} (q11)
          (q11) edge[above] node{\([4,\infty]\)} (q21)
          (q21) edge[loop below] node{\([1,3]\)} (q21)
          (q21) edge[above] node{\([4,\infty]\)} (q32)
          (q32) edge[below left] node{\([1,3]\)} (q42)
          (q32) edge[above] node{\([4,\infty]\)} (q43)
          (q42) edge[below right] node{\([4,\infty]\)} (q43)
          (q42) edge[loop below] node{\([1,3]\)} (q42)
          (q43) edge[above] node{\([1,\infty]\)} (q44)
          (q44) edge[loop right] node{\([1,\infty]\)} (q44)
          ;

\end{tikzpicture}
  \end{center}

\end{example}
Observing that transitions in the automaton are caused by intervals whose endpoints are associated to values in the composition leads us to the following definition.

\begin{definition}
  For any composition \(u\), an interval \([a, b-1]\) is \emph{critical} if one of
  the following holds:
  \begin{itemize}
  \item
  \( a = 1 \) and \( b \) is the minimum value of a non-1 symbol in \( u \),
  \item
  \( a \) and \( b \) are values of symbols in \( u \) and \( b \) is the least such value greater than \( a \), or
  \item
  \( a \) is the maximum value in \( u \), and \( b = \infty \).
  \end{itemize}
\end{definition}

\begin{example}
  The critical intervals of \(u = 373499\) are
  \begin{equation*}
    \left\{
      [1,2],[3,3],[4,6],[7,8],[9,\infty]
    \right\}
  \end{equation*}
\end{example}

In the automaton describing \( u \)-prolific permutations each transition is labeled by unions of critical intervals. This suggests that there
is a notion of standardisation relative to \( u \).

\begin{definition}
Given a composition, \(u\), we can define the \emph{ the \(u\)-standardisation
of \(w\)} as the function that takes each part of \(w\) and maps it to the
ordinal value of the critical interval \(u\) that contains the size of that
part.
\end{definition}
\begin{example}
  Suppose that \(u = 373499\) then the \(u\) standardisation of \(w= 8(12)4663281\)  is
  \begin{equation*}
    \std_u(w) = 453332141
  \end{equation*}
\end{example}

Note that the \(u\)-standardisation of \(u\) is supported by the set \(\nrange{m}\), or
\(\nrange{m+1} \setminus \{1\}\), where \(m\) is the number of distinct values of \(u\), and the latter occurs if
the symbol \(1\) does not occur in \(u\). It is the unique composition of this support such that the
order relations between corresponding elements are the same as in \(u\). This order preservation
property is similar to that of the notion of standardisation in permutations.

\begin{theorem}\label{thm:stdzation}
  For any composition \(u\)
  \begin{equation*}
    w \in \Pro(u) \iff \std_u(w) \in \Pro(\std_u(u))
  \end{equation*}
\end{theorem}
\begin{proof}
If two different values belong to the same critical interval for \( u \) then they can be exchanged for one another in any word \( v \) without
changing the supports of any occurrences of \( u \). In particular, they must induce the same transition from any state of the automaton
that recognises \( \Pro(u) \). But even more, all that matters is which critical interval they belong to in order from smallest to largest -- so
we can collapse each critical interval (except the unbounded one) to a single value which is the corresponding value in \( \std_u (u) \).
That is, except for this  relabelling, the two automata for \( \Pro(u) \) and \( \Pro( \std_u(u) ) \) are identical from which the result follows.

\end{proof}

\begin{example}
  Equivalence between the automaton for recognising \(1441\)-prolific compositions, and that recognising \(1221\)-prolific compositions (\(1221\) is the \(1441\)-standardisation of \(1441\)).
  
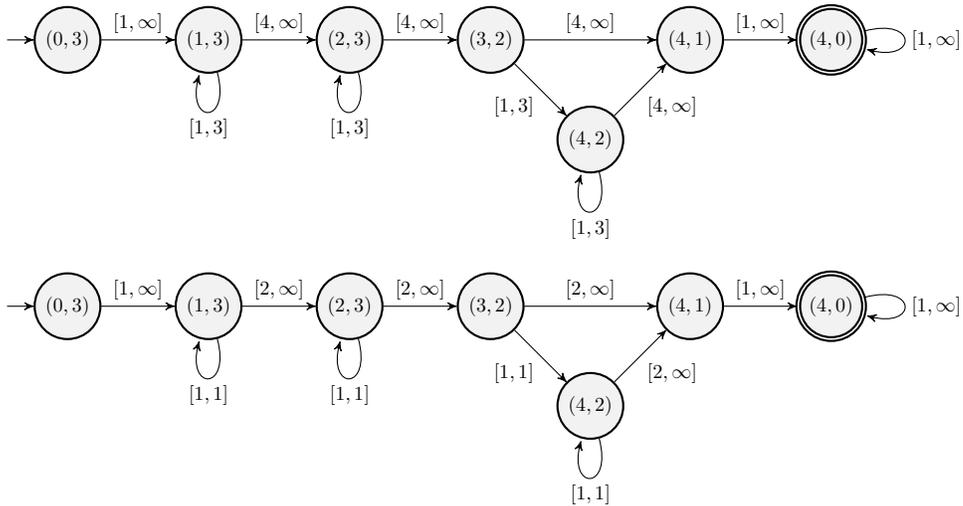
\begin{center}
\begin{tikzpicture}[scale=0.75, every node/.style={scale=0.75}]
    \node[state, initial] (q01) {\((0,3)\)};
    \node[state, right of=q01] (q11) {\((1,3)\)};
    \node[state, right of=q11] (q21) {\((2,3)\)};
    \node[state,  right of=q21] (q32) {\((3,2)\)};
    \node[state, below right of=q32] (q42) {\((4,2)\)};
    \node[state, above right of=q42] (q43) {\((4,1)\)};
    \node[state, accepting , right of=q43] (q44) {\((4,0)\)};

    \draw [->] (q01)  edge[above] node{\([1,\infty]\)} (q11)
          (q11) edge[loop below] node{\([1,3]\)} (q11)
          (q11) edge[above] node{\([4,\infty]\)} (q21)
          (q21) edge[loop below] node{\([1,3]\)} (q21)
          (q21) edge[above] node{\([4,\infty]\)} (q32)
          (q32) edge[below left] node{\([1,3]\)} (q42)
          (q32) edge[above] node{\([4,\infty]\)} (q43)
          (q42) edge[below right] node{\([4,\infty]\)} (q43)
          (q42) edge[loop below] node{\([1,3]\)} (q42)
          (q43) edge[above] node{\([1,\infty]\)} (q44)
          (q44) edge[loop right] node{\([1,\infty]\)} (q44)
          ;

\end{tikzpicture}
  \end{center}

    \begin{center}
\begin{tikzpicture}[scale=0.75, every node/.style={scale=0.75}]
    \node[state, initial] (q01) {\((0,3)\)};
    \node[state, right of=q01] (q11) {\((1,3)\)};
    \node[state, right of=q11] (q21) {\((2,3)\)};
    \node[state,  right of=q21] (q32) {\((3,2)\)};
    \node[state, below right of=q32] (q42) {\((4,2)\)};
    \node[state, above right of=q42] (q43) {\((4,1)\)};
    \node[state, accepting , right of=q43] (q44) {\((4,0)\)};

    \draw [->] (q01)  edge[above] node{\([1,\infty]\)} (q11)
          (q11) edge[loop below] node{\([1,1]\)} (q11)
          (q11) edge[above] node{\([2,\infty]\)} (q21)
          (q21) edge[loop below] node{\([1,1]\)} (q21)
          (q21) edge[above] node{\([2,\infty]\)} (q32)
          (q32) edge[below left] node{\([1,1]\)} (q42)
          (q32) edge[above] node{\([2,\infty]\)} (q43)
          (q42) edge[below right] node{\([2,\infty]\)} (q43)
          (q42) edge[loop below] node{\([1,1]\)} (q42)
          (q43) edge[above] node{\([1,\infty]\)} (q44)
          (q44) edge[loop right] node{\([1,\infty]\)} (q44)
          ;

\end{tikzpicture}
  \end{center}

\end{example}

\section{Minimal Prolific Compositions}\label{sec:MinProl}

The set of \(u\)-prolific compositions is closed upwards so it makes sense to try to determine its minimal elements. By a theorem of
\textcite{higmanOrderingDivisibilityAbstract1952}, the partial order of compositions is a quasi-well
order. Therefore the set of minimal \( u \)-prolific permutations will always be finite. In this section we
give a variety of results about these compositions.

\begin{theorem}
\label{thm:prepend-one}
  If \(\gamma\) is a minimally prolific composition for \(u \), then \( 1\gamma \) is minimally prolific for \(1 u \).
\end{theorem}
\begin{proof}

Consider the automaton that recognises prolific words for \( 1 u \). The initial suffix condition is \( 1 u \). \emph{Any} character
now changes the suffix condition to \( u \) and extends the prefix seen. But then the remainder of the automaton is exactly the automaton for recognising prolific words
for \( u \). So the minimum words accepted by this automaton consist of the minimum possible character to trigger the first transition
followed by a minimal word of the \( u \)-prolific automaton, exactly as claimed.
\end{proof}

\begin{corollary}
  If \(u\) is a minimally prolific composition for \(v\), then the composition
  \(\underbrace{1\dots1}_{n\text{ times}}u\) is minimally prolific for
  \(\underbrace{1\dots1}_{n\text{ times}}v\).
\end{corollary}

\begin{lemma}
  Given a composition \(u\), the unique minimal \(u\)-prolific composition is \(u\) itself if and only if \(u\) has a part of size one between every pair of parts of size greater than one.
\end{lemma}
\begin{proof}
  Under the given conditions on \(u\) in the Corollary to \cref{thm:coveringprolific} the second
  condition is vacuous. Since \(u\) is covered by itself and no smaller composition can be covered
  by copies of \(u\) the composition \(u\) is minimally prolific for itself.

  Suppose that \(u\) has a pair of parts with no one between them, then write  \(u\) as \(\alpha x y \beta\) with \(x, y >1\) then clearly \(\alpha x 1 y \beta\) contains no new copies of \(u\).
\end{proof}

If \( u \) is \( u \)-prolific then we call it a \emph{self-prolific} composition.

\begin{theorem}\label{thm:minconst}
    Given a composition of \(k+2\) parts, \(u = 1e_1e_2\ldots{}e_{k-1}e_k1\) such that for all \(i
    >0, e_i > 1\). There is a unique minimal \(u\)-prolific composition \(p\) given by: \[p =
    1e_1e_2\ldots{}e_{k-1}\max(e_k,e_1)\,e_2\ldots{}e_{k-1}e_k1\]
\end{theorem}

\begin{proof}
\noindent\emph{\(p\) is \(u\)-prolific.}

The composition \(p\) is covered by \(u\) since the first  \((k+1)\) parts form an occurrence of
\(u\) with any other part of \(p\).

 Adding a new part of size \(1\) anywhere before the \((k+1)\)\nst{} part creates a new occurrence of \(u\)
 with the last \((k+1)\) parts, and adding a new part of size \(1\) anywhere after the \((k+1)\)\nst{} part
 creates a new occurrence of \(u\) with the first \((k+1)\) parts. Therefore \(p\) is
 \(u\)-prolific.

\noindent\emph{No composition contained in \(p\) is \(u\)-prolific.}

Without loss of generality (due to symmetry under reversal) we can consider reducing the size of one
of the first \(k+1\) parts.

We have two cases:
\begin{case}[\(e_k \ge e_1\)] \label{case:c1}
  Consider inserting a new part of size 1 after the middle term. To the right of this part there are
  only \(k\) parts so we cannot find an occurrence of \(u\) in which the new part plays the role of
  the leftmost part of size 1. In order that we should have a new occurrence of \( u \) to the left of
  the new part would require each such part to match into the corresponding part of \( u \) since
  there are exactly enough parts available. But one of them is smaller than the corresponding part
  of \( u \) so this cannot happen. Therefore, no new occurrence of \( u \) is created.
\end{case}
\begin{case}[\(e_1 > e_k\)]
  If we have reduced any of the first \(k\) parts then we can insert a new part after the middle
  term and the logic follows from the previous case. On the other hand if we reduce the middle part
  then inserting a new part of size \(1\) before it does not create any new occurrences of \(u\).
  The new part cannot play the role of a rightmost one in any occurrence of \(u\) since there are
  only \(k\) parts preceding it. Also, it is not possible to use the new part as a leftmost one
  since all the remaining parts would be required but the first of them is now smaller than \(e_1\).
\end{case}
Taking both of these cases together tells us that no composition contained in \(p\) is
\(u\)-prolific.

\noindent\emph{Any composition not containing \(p\) is not \(u\)-prolific.}

Suppose there existed a composition \(v\) that did not contain \(p\) but was \(u\)-prolific. Write
\(p\) as \(1W1\). Since \(v\) avoids \(p\), and occurrence matching can occur greedily, it is
possible to concatenate additional parts to \(v\) to produce a composition that contains \(1W\) but does
not contain \(p\). That composition will still be \( u \)-prolific since \( v \) was.
So without loss of generality, we may assume
that \(v\) contains \(1W\) but not \(1W1\). Now \(v\) has the form

  \begin{equation*}
    v=v_1v_2\ldots{}v_m
  \end{equation*}

We can greedily match parts of \(1W\) into \(v\), the matching must finish at the part \(v_m\) since
\(v\) does not contain \(1W1\),

\begin{equation}
  \begin{matrix}
    v_1 & v_2 & \phantom{e_1}\ldots{} & x & \ldots & y &\phantom{e_2}\ldots & v_m \\
    \rotatebox[origin=c]{270}{\(\ge\)}&&\rotatebox[origin=c]{270}{\(\ge\)}\phantom{\ldots}&\rotatebox[origin=c]{270}{\(\ge\)}& &\rotatebox[origin=c]{270}{\(\ge\)}&\rotatebox[origin=c]{270}{\(\ge\)}\phantom{\ldots}&\rotatebox[origin=c]{270}{\(\ge\)}\\
    1 & & e_1\ldots & e_{k-1} & & \max(e_1,e_k)& e_2 \ldots & e_k
  \end{matrix}
  \label{eq:greedymatch}
\end{equation}

Consider adding a new part of size \(1\) immediately before the part \(y\) above. This new part cannot be
the initial one in an occurrence of \(u\) as the first layer that the part \(e_1\) can
match into is the part \(y\), and then matching greedily would continue as before and ends when
the part \(e_k\) has been matched into the part \(v_m\) -- so no copy of \( u \) would be found.

If this new part were a right hand \(1\) in an occurrence of \(u\) then we must have seen a part of
size \(e_k\) between the parts \(x\) and \(y\). This means that \(e_k < \max(e_1,e_k)\). Therefore,
if \(e_k \ge e_1\) then \(v\) is not \(u\)-prolific with this part not creating any new occurrences
of \(u\).

In the case that \(e_k < e_1\), consider putting a new part of size \(1\) immediately after \(x\),
this cannot be a right hand \(1\) in an occurrence of \(u\), as the prefix of \(u\) up to this part
is \(\range{1}{e_{k-1}}\) . If it were a left hand \(1\) in an occurrence of \(u\) then the part
\(y\) is the first part that matches \(e_1\), and the greedy match continues as before from this point,
and no new occurrence of
\(u\) will be found. So \(v\) is not \(u\)-prolific.
\end{proof}

\begin{lemma}\label{lem:partuse}
  If \(v = 1\gamma1\) is a minimally prolific composition for \(u= 1\beta1\) then for every part,
  \(v_i\), of \(v\) there exists some insertion of an element into \(1\gamma1\) such that any new
  occurrence of \(u\) must use \(v_i\).
\end{lemma}
\begin{proof}
Otherwise every insertion of a 1 creates a new occurrence of \( u \) not using \( v_i \). But then the
composition formed by deleting \( v_i \) from \( v \) is already prolific, contradicting \( v \)'s minimality.
\end{proof}

\begin{theorem}
  If \(1\alpha1\) is self-prolific and \(1\gamma1\) is a minimally prolific composition for
  \(1\beta1\) then the composition \(v = 1\alpha1\gamma1\) is minimally prolific for \(u =
  1\alpha1\beta1\).
\end{theorem}

\begin{proof}
The argument is essentially the same as that for Theorem \ref{thm:prepend-one}. Because \( 1 \alpha 1 \) is self-prolific
the automaton for \( \Pro (1 \alpha 1 ) \) has loop transitions labelled by values less than the next element of \( 1 \alpha 1 \) and
proper transitions to the next state on larger values. It is a single chain with no branches.
But, in considering the automaton for \( \Pro (1 \alpha 1 \beta 1  ) \) the same
transitions apply only the suffix conditions have \( \beta 1 \) appended. Therefore the complete structure of the
 \( \Pro (1 \alpha 1 \beta 1  ) \) automaton is obtained from the \( \Pro (1 \beta 1) \) automaton by fusing its initial state
 (corresponding to an empty prefix) with the penultimate state of the \( \Pro (1 \alpha 1 ) \) automaton (corresponding to a remaining
 suffix requirement of 1.)

In particular the minimal compositions accepted by this automaton are exactly \( 1 \alpha \) followed by a minimal composition
accepted by the \( \Pro(1 \beta 1) \) automaton.
\end{proof}

We have now seen that certain compositions \(u\) have unique minimal \(u\)-prolific compositions.

Now we consider methods for determining the minimal \(u\)-prolific compositions for compositions that are not encompassed by the previous results.

An automaton, \(A\), that accepts minimally prolific compositions (and possibly some others) for the composition \(u\) is easily constructed from the automaton, \(D\) that
accepts all \(u\)-prolific words. Consider what happens if we were to follow a transition that loops in \(D\), if we then reach an accepting state we could find a smaller \(u\)-prolific composition by omitting the part introduced by this loop, therefore we can redirect any transitions that label loops into a non-accepting sink state.
Any interval causing a transition can be replaced by the smallest value it contains, and any larger values can also be redirected to the sink state since any composition obtained by using a larger value causing the same transition can be made smaller by using the smaller value.

\begin{example}
  The automaton for recognising minimal \(1221\)-prolific compositions, as well as some other \(1221\)-prolific compositions.
  
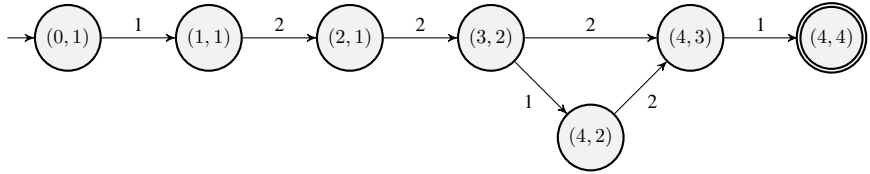
\begin{center}
\begin{tikzpicture}[scale=0.75, every node/.style={scale=0.75}]
    \node[state, initial] (q01) {\((0,1)\)};
    \node[state, right of=q01] (q11) {\((1,1)\)};
    \node[state, right of=q11] (q21) {\((2,1)\)};
    \node[state,  right of=q21] (q32) {\((3,2)\)};
    \node[state, below right of=q32] (q42) {\((4,2)\)};
    \node[state, above right of=q42] (q43) {\((4,3)\)};
    \node[state, accepting , right of=q43] (q44) {\((4,4)\)};

    \draw [->] (q01) edge[above] node{1} (q11)
          (q11) edge[above] node{2} (q21)
          (q21) edge[above] node{2} (q32)
          (q32) edge[below left] node{1} (q42)
          (q32) edge[above] node{2} (q43)
          (q42) edge[below right] node{2} (q43)
          (q43) edge[above] node{1} (q44)
          ;

\end{tikzpicture}
  \end{center}
\end{example}

After these modifications the automaton accepts fewer compositions but still accepts only \( u \)-prolific compositions and still accepts
all minimal \( u \)-prolific compositions. If there are no branches then in fact we will have established that there is only one minimal \( u \)-prolific permutation. However, if there are branches then Dijkstra's algorithm can be applied to find a minimal weight path to the accepting state -- where the weight of a path is the sum of its labels. This path must represent a minimal \(u\)-prolific composition, \( v \). But now we can easily modify the automaton to accept only those words accepted by the original one which do not contain \( v \) (simply by keeping track of what prefix of \( v \) has been found and passing to a sink state if we contain \( v \)). If this modification still has accepting computations then we can find a minimum weight word that it accepts which is another minimal \( u \)-prolific composition. Since we know that the set of minimal \( u \)-prolific permutations is finite this procedure terminates with the complete set.

Application of this technique has allowed us to verify the theory
presented in this section as well as find some examples of
interesting behaviour that leads to further questions and avenues for investigation. For example, in the family of compositions having the structure

\begin{equation*}
  M_k = 122\underbrace{1\dots1}_{k\text{ ones}}221
\end{equation*}

\noindent
there are \(k+1\) minimally prolific compositions for the composition \(M_k\) for all \(k\) from \(0\) to \(8\). This gives rise to the following conjecture.

\begin{conj}
  Given any integer \(r\) there exists a composition \(u\) such that the set of minimally \(u\)-prolific compositions has exactly \(r\) elements.
\end{conj}

\section{Future}

The concept of prolificity can be framed in another way. Consider any class of combinatorial objects that carries
a notion of embedding where each embedding of a structure \( A \) into a structure \( B \) is witnessed by
an injective map from \( A \) to \( B \). Then, even if we do not want to count embeddings, it makes sense to say
that \( B \) is covered by \( A \) if the union of the ranges of such maps is \( B \). Now we can say that \( B \) is
\emph{\( A \)-prolific} if every structure containing \( B \) is covered by \( A \). In compositions this definition is
equivalent since we obtained covering as a necessary condition (strictly speaking the counting definition would require
only that for a structure \( C \) containing \( B \) the union of the images of \( A \) should contain \( C \setminus B \)).

Despite this natural framing it seems that this concept has not previously been investigated extensively. We believe that
the previous sections show that there are contexts in which it is of interest -- at the very least for integer
compositions, where a number of open questions still remain. For instance, aside from the two conjectures above we can ask: if \( u \)
is a composition with \( k \) parts, then what is the maximum possible size of the set of minimal \( u \)-prolific compositions?
Note this question does have a finite answer since we may assume that \( u \) is standardised and there are only finitely
many standardised compositions with \( k \) parts.

We can easily consider prolificity in other combinatorial structures, for instance graphs.
If embedding is taken to be as an induced subgraph, it is
easy to see that there are no \( G \)-prolific graphs for any graph \( G \) other than a single vertex. This is because a graph can be extended by adding a new vertex independent of all others, and also by adding a new vertex adjacent to all others. For
\( G \) having more than one vertex, one of these extensions fails to contain new copies of \( G \). However, it is easy to establish that, in the class of graphs of bounded degree \( d \), there exist graphs that are prolific for the graph that has three vertices and one edge.

In general we believe that the use of automata in considering prolific structures, as we were able to do for compositions, is a powerful tool. This suggests that further investigations should concentrate on classes that allow the representation of structures as words over some alphabet.

\printbibliography
\end{document}